\documentclass{article}

\usepackage{microtype}
\usepackage{graphicx}
\usepackage{booktabs} 

\usepackage{amsmath}
\usepackage{amsthm}
\usepackage{amsfonts}
\usepackage{amssymb}
\usepackage{nccmath}
\usepackage{bbm}

\usepackage{comment}
\usepackage{caption}
\usepackage{subcaption}

\usepackage{algorithm}
\usepackage{algorithmic}

\usepackage{booktabs}
\usepackage{longtable}

\usepackage{bm}

\newtheorem{thm}{Theorem}[section]

\newtheorem{lem}[thm]{Lemma}

\newtheorem{prop}[thm]{Proposition}

\newcommand{\RR}{ \mathbb{R} }

\newcommand{\spaceo}{\hspace{2 mm}}
\newcommand{\setsep}{ \spaceo | \spaceo}

\newcommand{\argmax}{\operatornamewithlimits{argmax}}
\newcommand{\argmin}{\operatornamewithlimits{argmin}}

\newcommand{\Abs}[1]{\left| #1 \right|}
\newcommand{\Set}[1]{\left\{ #1 \right\}}
\newcommand{\Brack}[1]{\left( #1 \right)}
\newcommand{\BBrack}[1]{\left\{ #1 \right\}}

\newcommand{\Exp}[1]{ \mathbb{E} #1}
\newcommand{\Expsubidx}[2]{ \mathbb{E}_{#1} #2}
\newcommand{\norm}[1]{\left\|#1\right\|}

\newcommand{\Ind}[1]{ \mathbbm{1}_{\Set{#1}} }
\newcommand{\eps}{\varepsilon}

\renewcommand{\algorithmicrequire}{\textbf{Input:}}

\newcommand{\whM}{\widehat{M}}

\usepackage{hyperref}

\newif\ifimagesshow
\imagesshowfalse

\title{Topic Modeling via Full Dependence Mixtures}

\author{
Dan Fisher$^1$ \\
{\tt\small danf@technion.ac.il} \\
\and
Mark Kozdoba$^1$ \\
{\tt\small markk@technion.ac.il}
\and
Shie Mannor$^1$ \\
{\tt\small shie@ee.technion.ac.il}
}
\date{
     $^1${\tt\small Technion, Israel Institute of Technology}
}

\begin{document}
\maketitle

\begin{abstract}
In this paper we introduce a new approach to topic modelling that scales to 
large datasets by  using  a compact representation of the data and by
leveraging the GPU architecture.   
In this approach, topics are learned directly from the  
co-occurrence data of the corpus. In particular, we introduce a novel
mixture model which we term the Full Dependence Mixture (FDM) model.
FDMs model  second moment under  general generative 
assumptions on the data. While there is previous work on topic 
modeling using second moments,  we develop a direct stochastic 
optimization procedure for fitting an FDM with a single Kullback 
Leibler objective. Moment methods in general have the benefit that 
an iteration no longer needs to scale with the size of the corpus. 
Our approach allows us to leverage standard 
optimizers and GPUs for the problem of topic modeling. In 
particular, we evaluate the approach on two large datasets, 
NeurIPS papers and a Twitter corpus, with a large number of 
topics, and show that the approach performs comparably or better than the the standard benchmarks.
\end{abstract}

\section{Introduction}
\label{sec:intro}
A topic model is a probabilistic model of joint distribution in the data, that is typically used as a dimensionality reduction
technique in a variety of applications, such as for instance text mining, information retrieval and recommender systems.  
In this paper we concentrate on topic models in  text data.   Perhaps the most widely used topic model for text is
the Latent Dirichlet Allocation (LDA) model, 
\cite{Blei02latentdirichlet}. LDA is a fully generative probabilistic model, and is typically learned through a Bayesian approach --- by sampling the parameter  distribution given the data, via Collapsed Gibbs samplers \cite{cgibbs}, \cite{cgibbs_book}, variational methods, \cite{Blei02latentdirichlet},  
\cite{variational1}, \cite{variational2}, 
or other related approaches. A typical  
learning procedure for Bayesian methods involves an iteration over the entire corpus, where a topic assignment is sampled per each token or document in the corpus. In order to apply these methods to large corpora, a variety of optimized procedures were developed, where speed improvement is achieved either via parallelization or via more economic sampling scheme, per token. An additional level of complexity is added by the fact that LDA has two hyperparameters, the Dirichlet priors of the token distribution in topics, and topic distribution in documents. This may complicate the application of the methods since the choice of parameters may influence the results even for relatively large data sizes.  A detailed discussion of LDA optimization is given in Section \ref{sec:literature}. 

In this paper we propose an alternative approach to topic 
modeling, based on principles that are principally different from the standard LDA optimization techniques. We 
show that using this approach it is possible to 
analyze large datasets on a single 
workstation with a GPU, and to obtain 
results that are comparable or better 
than the standard benchmarks. 

In the rest of this section we introduce some necessary notation, describe our model and the related loss function, discuss the optimization  procedure for this loss, and  overview the experimental results of the paper.

\subsection{Method Description}
\label{sec:subsec_method_desc}
We assume that the text data was generated from a pLSA (Probabilistic Latent Semantic Allocation) probability model , \cite{PLSA}, as follows: Denote by $\mathcal{X}$ the set of distinct tokens in the corpus, 
and suppose that we are given $T$ topics, $\mu_{t}$, $t\leq T$, where each topic is a probability distribution on the set of tokens $\mathcal{X}$. 
Then the pLSA assumption is that each document $d$ is 
generated by independently sampling tokens from a mixture of topics, denoted $\nu_d$:
\begin{equation}
    \label{eq:intro_mu_d}
    \nu_d = \sum_t \theta_d(t) \mu_t,
\end{equation}
where $\theta_d(t)\geq 0$ and $\sum_{t} \theta_d(t) = 1$ for every document 
$d$.  Note that we do not specify the generative model for $\theta_d$. 
In this sense, pLSA is a semi-generative model, and is more general than LDA.

Next, for every document in the corpus we construct its token co-occurrence 
probability matrix, and we take the 
co-occurrence probability matrix of the corpus to be the average of the 
document matrices. Let $N = \Abs{\mathcal{X}}$ be the dictionary size - 
the number of distinct tokens in the corpus. Then the co-occurrence matrix 
$\widehat{M}$ of the corpus is an $N\times N$ matrix, with 
non-negative entries that sum to $1$. Suppose that one performs the 
following experiment: Sample a document from a 
corpus at random, and then sample two tokens independently from the 
document. Then $\widehat{M}_{u,v}$ is the probability 
to observe the pair $u,v$ in this experiment (up to a small modification, 
full details are given in Section \ref{sec:methods}).

Now, if one assumes the pLSA model of the corpus, then it can be shown that 
the expectation of $\widehat{M}$ should be of the form 
\begin{equation}
    \label{eq:fdm_def}
    M_{u,v}(\mu,\alpha) = \sum_{i,j=1}^T \alpha_{i,j} \mu_i(u) \mu_j(v), 
\end{equation}
where $\mu_i$ are the topics and $\alpha_{i,j}\geq 0$, $\alpha_{i,j}=\alpha_{j,i}$, and  $\sum_{i,j} \alpha_{i,j} = 1$ represent the corpus level topic-topic correlations. 
We refer to the matrices of the form (\ref{eq:fdm_def}) as Full Dependence 
Mixture (FDM) matrices. This is due to the analogy with standard 
multinomial mixture models (also known as categorical mixture models), which can be represented in the form 
(\ref{eq:fdm_def}) but with $\alpha$ restricted to be a diagonal matrix. Multinomial mixture models correspond to the special case 
where each document contains samples from only one topic, where the topic may be different for different documents. 
Equivalently, in multinomial mixtures, each $\theta_d$ is $0$ on all but one coordinates. 

In this paper, we consider a set of topics $\mu_t$ to be a good fit for the 
data if there are some correlation coefficients $\alpha$ such that 
$M(\mu,\alpha)$ is close to $\widehat{M}$, the FDM generated from the data. 
Specifically, we define the loss by 
\begin{equation}
 \label{eq:L_defin}
    L(\mu, \alpha) = -\sum_{u,v \in \mathcal{X}} \widehat{M}_{u,v} \log M(\mu,\alpha)_{u,v} , 
\end{equation}
and we are interested in minimizing $L$ over all $\mu,\alpha$. 
Clearly, minimizing  $L$ is equivalent to minimizing the Kullback-Leibler 
divergence between $M(\mu,\alpha)$ and $M$, viewed as probability 
distributions over $\mathcal{X} \times \mathcal{X}$. 

\begin{figure}
\centering
\includegraphics[width=\linewidth]{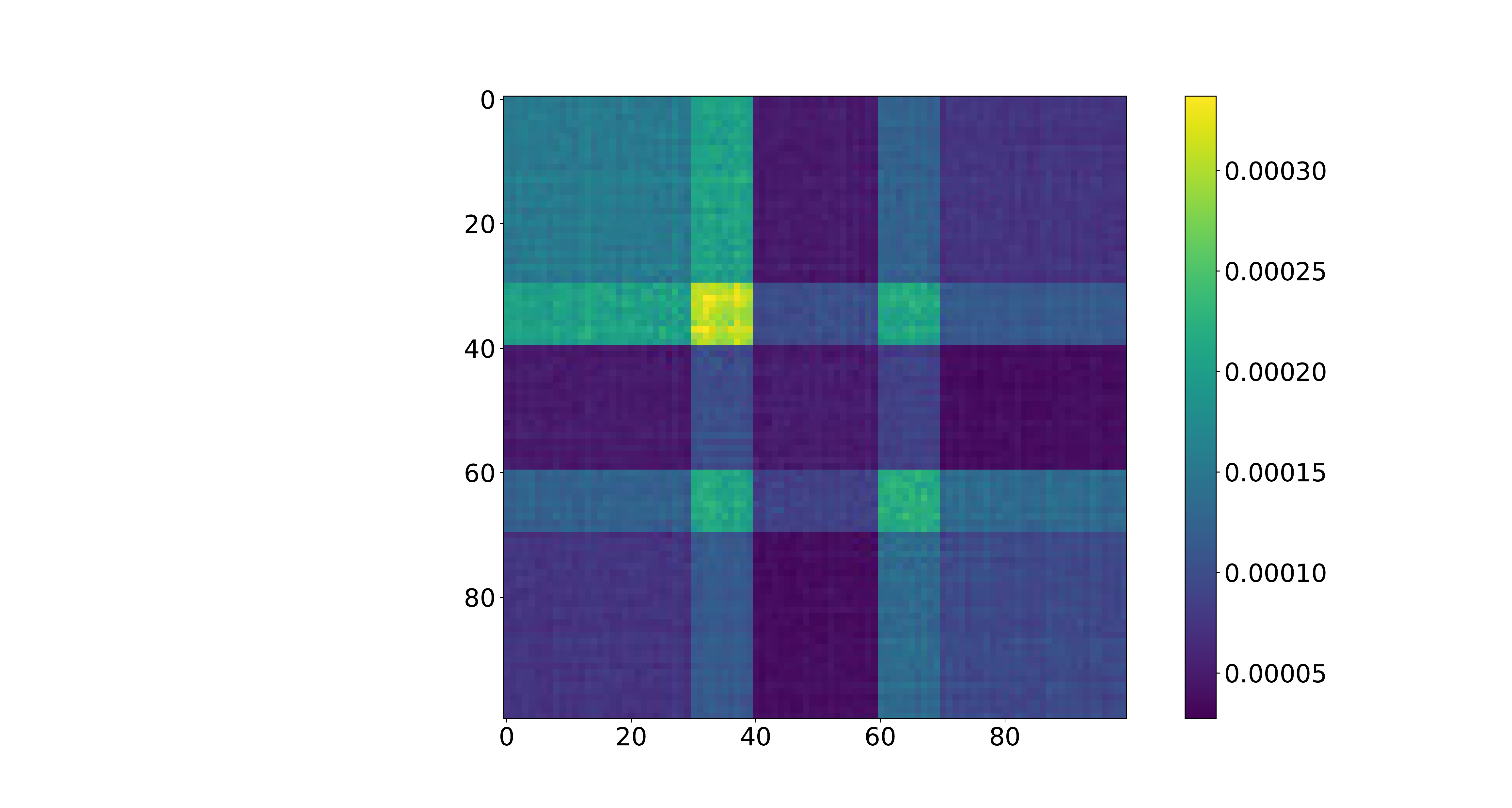}
\caption{An $\whM$ matrix example.}
\label{fig:sample_M}
\end{figure}

To gain basic intuition as to why $M(\mu,\alpha)$ that approximates well the empirical matrix $\whM$ should yield
good topics, it is useful to consider again the simple case of the multinomial mixture, and moreover, the specific case
where the topics are disjoint. In this case, the matrix $\whM$ will be block diagonal (up to a reordering of the dictionary), with disjoint blocks that correspond to the 
topics. Thus, provided enough samples $d$, the topics can be easily read off from the matrix. 
An example of a more complicated matrix $\whM$ is shown in Figure \ref{fig:sample_M}. Here the ground set is $\mathcal{X} = \Set{1,\ldots,100}$, and the topics are uniform on intervals $[1,40]$,$[30,70]$, and $[60,100]$ respectively. The documents were generated from the mixture (\ref{eq:intro_mu_d}), with 
$\theta_d$ sampled from a non-uniform Dirichlet, $\theta_d \sim Dir((2,1,1.5))$, which make the topics appear with different frequencies in the data. Although here the relation between the topics and the matrix is more involved, one can still see that the topics could be traceable from the matrix.

In Section \ref{sec:consistency} we show the asymptotic consistency of the loss (\ref{eq:L_defin}) for topics which satisfy the classical anchor words (\cite{anchor_words_nmf},\cite{Arora2}) assumption. That is, if the topics satisfy the anchor words assumption, then given enough samples the topics can be uniquely reconstructed from  the matrix $\whM$ by minimizing the loss (\ref{eq:L_defin}).  The anchor words assumption roughly states that each topic has a word that is unique to that particular topic. Note that this word does not have to appear in each document containing the topic, and may in fact have a relatively low probability inside the topic itself. It is known, \cite{Arora2}, and easy to verify, that natural topics, such as topics produced by learning LDA, do satisfy the anchor words assumption. We refer to Section \ref{sec:consistency} for further details.

The advantage of using the cost $L(\mu,\alpha)$ is that it depends 
on the corpus {\em only} through the matrix $\widehat{M}$. Therefore, the size of 
the corpus does not enter the optimization problem directly, and 
we are dealing with a fixed size, $N\times N$ problem. This is a 
general feature of reconstruction through moments approaches, such as for instance \cite{Arora2}, \cite{tensor3}. 
In particular, the number of variables for the optimization is $TN + T^2$, in contrast to variational or Gibbs sampler based 
methods, which either have $T$ variables for every document, or have a single variable for every token in the corpus, 
respectively. 

\subsection{Optimization of the Objective $L(\mu,\alpha)$.}
\label{sec:subsec_optimization}
For smaller problems, one may directly optimize the objective 
(\ref{eq:L_defin}) using gradient descent methods. However, note that 
if one computes $L(\mu,\alpha)$ directly, then one has to compute  
$M(\mu,\alpha)$, which is a sum of $T^2$ matrices of size $N^2$. Indeed, 
denote by $M^{i,j}$ the $N\times N$ matrix such that its $u,v$ entry is $M^{i,j}_{u,v} = \mu_i(u) \cdot \mu_j(v)$. Then 
$M(\mu,\alpha) = \sum_{i,j\leq T} \alpha_{i,j} M^{i,j}$.
On standard GPU computing architectures, all $T^2$ of the matrices 
will have to be in memory simultaneously, which is prohibitive even for 
moderate values of $N,T$. To resolve this issue, we reformulate 
the optimization of $L$ as a stochastic optimization problem in $u,v$. 
To this end, note that $L$ is an expectation of the term 
$\log M(\mu,\alpha)_{u,v}$ over pairs of tokens $(u,v)$, sampled from 
$\widehat{M}$, where $\whM$ is viewed as a probability distribution over $\mathcal{X} \times \mathcal{X}$. Formally, 
\begin{equation}
    L(\mu,\alpha) = \Expsubidx{(u,v) \sim \widehat{M}} \log M(\mu,\alpha)_{u,v}.
\end{equation}
Therefore, given $(u,v)$, one only has to compute the gradient of  
$M(\mu,\alpha)_{u,v}$, rather than full
$M(\mu,\alpha)$ at $\mu,\alpha$ -- which is a much smaller 
optimization problem, of size $O(T^2)$,
and this can be done for moderate $(u,v)$ batch sizes.  This approach makes 
the optimization of $L(\mu,\alpha)$ practically feasible.  The full algorithm, given as  Algorithm \ref{alg:fdm_opt} below, is discussed in detail in Section \ref{sec:methods}. Note that this approach differs from the standard stochastic gradient descent flow on the GPU, where
the batches consist of data samples (documents in the case of text data). Instead, here the data is already summarized as $\widehat{M}$, and the batches consist of pairs of tokens $(u,v)$ that we sample actively from  $\whM$.

\subsection{Experimental Results}
\label{sec:subsec_exprimental}
We evaluate the FDM algorithm on a semi-synthetic dataset where the ground truth topics are known and are taken to be realistic (topics learned by LDA on NeurIPS data) and on two real world datasets: the NeurIPS  full papers corpus, and a very large (20 million tweets) Twitter dataset that was collected via the Twitter API.  For the semi-synthetic dataset the topic quality was measured by comparison to the ground truth topics, while for the real datasets log-likelihhod on a hold-out set was measured. We compare FDM to a state of the art LDA collapsed Gibbs sampler (termed SparseLDA), and to the topic learning algorithm introduced in \cite{Arora2} (termed EP).  Additional details on these benchmarks are given in Section \ref{sec:literature}. For the semi-synthetic models, we find that while, as expected, SparseLDA with true hyperparameters performs somewhat better, FDM performs similarly to a SparseLDA with close but different hyperparameters. All algorithms find a reasonable approximation of the topics, althpough EP performance is notably weaker. On NeurIPS FDM performs similarly to SparseLDA, while on Twitter data FDM performance is somewhat better. Both algorithms outperform EP.

To summarize, the contributions of this paper are as follows: We introduce a new  approach to
topic modeling, via the fitting of the empirical FDM $\widehat{M}$ to the  
topic FDM $M(\mu,\alpha)$ by likelihood minimization, and prove the consistency of the associated loss. We introduce a practical optimisation procedure for this method, and experimentally show that it produces topics that are comparable or better than the state of the art approaches, while using  principles that significantly differ from the existing methods.

\section{Literature}
\label{sec:literature}


The subject of optimization in topic models has received significant attention in the literature. 
We first describe the general directions of this research. Variational methods for the LDA objective were 
developed in the paper that introduced the LDA model, \cite{Blei02latentdirichlet}. See also 
\cite{variational1}, \cite{variational2}. The collapsed Gibbs sampler for LDA was introduced in 
\cite{cgibbs}, \cite{cgibbs_book}. Optimizations of the collapsed sampler that exploit the sparsity 
of the topics in a document were developed in \cite{mallet2}, \cite{cgibbs_sparcity}, \cite{walker_alias} and yielded further 
  performance improvements. 
Streaming, or online methods for the LDA objective were proposed in \cite{mallet1}, \cite{streaming_lda}. 
We also note that the topic modelling problem, and in particular the pLSA model, is closely related to the Non-Negative Matrix Factorization (NMF) 
problem, \cite{nmf_survey}. In this context, EM type algorithms for topic models were developed in the 
paper that introduced pLSA, \cite{PLSA}. Streaming algorithms for NMF in general are also an active 
field, see for instance \cite{online1}, \cite{online2}, which involve Euclidean costs, but could possibly be adapted to the topic modelling setting.   Finally, distributed methods for 
LDA were introduced in \cite{mallet1}, \cite{smola_parallel}, \cite{parallel2}, 
\cite{parallel3}, \cite{parallel4}.

Topic reconstruction from corpus statistics such as the matrix 
$\widehat{M}$ were previously considered in the theoretical study of topic models. Topic reconstruction
from the third moments of the data was proposed in \cite{tensor3}. While highly theoretically significant, these algorithms 
require construction an analysis of an $N\times N \times N$ matrix, and thus can not be applied with dictionaries of size 
of several thousands tokens. An algorithm 
that is based on the matrix $\widehat{M}$, as in our approach, was given in \cite{Arora1} and 
 improved in \cite{Arora2}. 
However, despite the fact that both \cite{Arora2} and our approach use  $\widehat{M}$, the underling principles behind the two algorithms are completely different. The approach of \cite{Arora2} is based on the notion of anchor words (see Section \ref{sec:consistency}), and consists of two steps: First, the algorithm attempts to explicitly identify the anchor words from the matrix $\whM$. Then, the topics are reconstructed using the obtained anchor words. 
Due to the structure of the pLSA model, it can be shown that any raw of the matrix $\whM$ can be approximately represented as a convex combination of the raws of $\whM$ that correspond to the anchor words. Thus, the anchor word identification in \cite{Arora2} is done by identifying the approximate extreme points of a set, where the set is the set of raws of $\whM$. In contrast, our approach does not attempt to find or use anchor words explicitly and conceptually is a much simpler gradient descent algorithm. We optimize in the space of topics directly, by approximating $\whM$ by an FDM $M(\mu,\alpha)$. It is also worth mentioning that in the consistency proof of Algorithm \ref{alg:fdm_opt}, Theorem \ref{thm:consistency} (although not in the algorithm itself), we use the characterization of the topics as the extreme points of a certain polytope. However, these are not the same extreme points as in \cite{Arora2}. The extreme points we use for the proof correspond to topics, while the extreme points in \cite{Arora2} correspond to conditional probabilities of tokens given anchor words, and are generally very different from the topics themselves. Finally, as mentioned earlier, we use the algorithm form \cite{Arora2} as an additional benchmark. While the algorithm of \cite{Arora2} is extremely fast, and can be very precise under certain conditions, the quality of the topics found by that algorithm is significantly lower compared to the topics produced by the standard optimized LDA Gibbs sampler.

A variant of the Gibbs sampler for LDA that is adapted to the computation on GPU was recently proposed in \cite{tristan}. We note that similarly to the Gibbs samplers or variational methods, this approach  
maintains a form of a topic assignment for each document. Therefore, the number of variables that need to be stored grows
with the number of documents \emph{and} topics and is $\Omega(D\cdot T)$, where $D$ is the number of documents and $T$ number of topics. This is true despite the remarkable memory optimizations described in \cite{tristan}, which address other matrices used by that algorithm. Observe that for GPUs, this problem is quite severe, since 
the amount of memory typically available on a GPU is much smaller than the CPU memory. This is in contrast to our approach, where 
the GPU memory requirement is independent of the number of documents or tokens.
To put this in context, the datasets we analyze here, NeurIPS and Twitter (both with $T=500$) can not be analyzed via the approach of \cite{tristan} on a high end desktop GPU ($10GB$ memory).

The MALLET code, \cite{McCallumMALLET}, is widely used as the standard benchmark.  This code is based on a collapsed Gibbs sampler for LDA, and implements a variety of optimizations discussed earlier. In particular it exploits sparsity, based on \cite{mallet2},  parallelization (within a single workstation) based on 
\cite{mallet1}, and has an efficient and publicly available implementation.

Finally, we note that, while outside of the scope of this paper, the methods presented here 
could easily be adapted to distributed, multi-GPU settings, using standard SGD parallelization techniques. 
This may be achieved either by elementary means, by increasing the batch size and processing it on 
multiple GPUs in parallel, or via more involved, lock-free methods, such as \cite{hogwild}.

\section{Formal Algorithm Specification}
\label{sec:methods}
Recall that $\Abs{\mathcal{X}}=N$ is the size of the dictionary. For a document $d$ given as a sequence 
of tokens $d = \Set{x_1,\ldots,x_{l_d}}$, where $l_d$ is the total number of tokens in $d$, 
denote by $c_d \in \RR^N$ the count vector of $d$, 

\begin{equation}
\label{eq:c_d_def}
c_d(u) = \#\Set{ x_i \setsep x_i = u } \mbox{ for $u \in \mathcal{X}$}. 
\end{equation}
Thus, $c_d$ is the bag of words representation of $d$,  and $c_d(u)$ is 
the number of times $u$ appears in $d$.  With this notation, the 
construction of the matrix $\whM$ from the data is described in 
Algorithm \ref{alg:fdm_calc}.

To motivate this construction, assuming the pLSA model, suppose a document $d$ is sampled from a mixture of topics
\begin{equation}
\nu = \sum_{t\leq T} \theta(t) \mu_t.
\end{equation}
The co-occurrence matrix of the mixture $\nu$ is 
\begin{equation}
\label{eq:m_nu}
    \Brack{M_{\nu}}_{u,v} = \nu(u)\cdot \nu(v) = \sum_{i,j\leq T} \theta(i)\theta(j) \mu_i(u) \mu_j(v).
\end{equation}
Thus, $\Brack{M_{\nu}}_{u,v}$ is the 
probability of obtaining the pair 
$(u,v)$ when one samples two token 
independently from $\nu$. The 
co-occurrence matrix of the corpus is 
the average of the corresponding $M_{\nu}$ over all 
documents $d$. Observe form 
(\ref{eq:m_nu}) that $M_{\nu}$ that has the form of 
an FDM, (\ref{eq:fdm_def}), and 
hence the co-occurrence matrix of the full corpus also has this form.  Next, note that we do not have access to the matrices $M_{\nu}$ themselves, only to the documents $d$, which are samples from $\nu$. We therefore estimate $M_{\nu}$ using the tokens of the document. Specifically, the matrix $\whM_d$ constructed in (\ref{eq:whM_d_alg_definition})   in Algorithm \ref{alg:fdm_calc} is an unbiased estimate of $M_{\nu}$ from the tokens in $d$: We have
\begin{equation}
\label{eq:unbiased_estimate_main_text}
    \Expsubidx{d}{\whM_d \; | \theta} = M_{\nu},
\end{equation}
where the expectation is over the documents sampled from $\nu$. We provide the proof of this 
statement in the supplementary material. Therefore, to obtain an approximation to the  
co-occurrence matrix of the model, in  Algorithm \ref{alg:fdm_calc} we 
first compute the estimates $\whM_d$ for each document, and then average 
over the corpus.

\begin{algorithm}[tb]
   \caption{Computation of $\whM$}
   \label{alg:fdm_calc}
\begin{algorithmic}[1]
    \REQUIRE Corpus: $\mathcal{C} = \Set{d_1,...,d_D}$, a corpus with $D$ documents.  
    \STATE For every $d \in \mathcal{C}$ construct $\whM_d$ such that the entry 
           $u,v$ is:
\begin{eqnarray}
\Brack{\whM_d}_{u,v} = \begin{cases}
\frac{l_d}{l_d -1} \frac{c_d(u)}{l_d}\frac{c_d(v)}{l_d} &  \text{ if $u \neq v$} \\
\frac{l_d}{l_d -1} \frac{c_d(u)}{l_d}\frac{c_d(v)}{l_d} - \frac{c_d(u)}{l_d\Brack{l_d-1}} & \text{ if $u = v$}
\end{cases}.
\label{eq:whM_d_alg_definition}
\end{eqnarray}
    \STATE Set 
    \begin{equation}
    \label{eq:whm_formal_def}
    \whM = \frac{1}{D} \sum_{d \in \mathcal{C}} \whM_d. 
    \end{equation}
\end{algorithmic}
\end{algorithm}

\begin{algorithm}[tb]
   \caption{FDM Optimization}
   \label{alg:fdm_opt}
\begin{algorithmic}[1]    
    \REQUIRE $\widehat{M}$: The empirical FDM
    \REQUIRE $B$: Batch size
    \REQUIRE $T$: Number of topics
    \STATE Initialize free variables $\alpha'$, $\mu_t'$ at random. 
    \STATE Set $\alpha = softmax(\alpha'), \mu_t = softmax(\mu_t')$.
    \WHILE{$L_B$ not converged}
        \STATE Sample $B$ pairs of tokens, $Batch = \Set{(u_1,v_1),...,(u_B,v_B)}$. Each pair is sampled from $\widehat{M}$, $(u_i,v_i) \sim \widehat{M}$. 
        \STATE Set $L_B(\mu,\alpha) = \sum_{k=1}^{B}\log\sum_{i,j}^{T}\alpha_{i,j}\mu_{i}(u_k)\mu_{j}(v_k)$
        \STATE $(\mu',\alpha') \gets  (\mu',\alpha') + \gamma \nabla_{\mu',\alpha'} L_B$
    \ENDWHILE    
\end{algorithmic}
\end{algorithm}

Once the matrix $\whM$ is constructed, our goal is to reconstruct the topics $\mu$ and the corpus level coefficients $\alpha$ from $\whM$. As discussed in the introduction, the FDM optimization algorithm 
is a stochastic gradient ascent on the pairs of tokens $(u,v)$ sampled from $\whM$, given as
 Algorithm \ref{alg:fdm_opt}.  Note that since the parameters
 $\alpha,\mu$ in which we are interested are constrained to be probability distributions, we parametrize them 
 with free variables $\alpha',\mu_t'$ via softmax:
 \begin{equation}
     \alpha_{i,j} = \frac{e^{\alpha'_{ij}}}{\sum_{i',j'\leq T} e^{\alpha'_{i'j'}}} \mbox{ and }
     \mu_t (l) = \frac{e^{\mu_t'(l)}}{\sum_{l' \leq N} e^{\mu_t'(l')}},
 \end{equation}
 where $i,j \leq T$, $t\leq T$,  and $l\leq N$. Thus, in Algorithm \ref{alg:fdm_opt}, $\alpha,\mu$ are functions of $\alpha',\mu'$ and the 
 gradient ascent is over $\alpha',\mu'$. Note that any SGD optimizer, including adaptive rate optimizers, may be used in Step $6$ of Algorithm \ref{alg:fdm_opt}. We use Adam, \cite{adam}, in the experiments.
 

\section{Consistency}
\label{sec:consistency}
\newcommand{\nsimplx}{\Delta_{N-1}}
\newcommand{\tsimplx}{\Delta_{T-1}}
In this section we discuss the consistency of the loss function (\ref{eq:L_defin}) under the anchor words assumption. Specifically, we show that if the the corpus is generated by a pLSA model with a set of topics $\Set{\mu_t}_{1}^T$ that satisfy the anchor words assumption, then among topics satisfying this assumption, the loss 
(\ref{eq:L_defin}) can only be asymptotically  minimized by an FDM $M(\mu,\alpha)$ with the true topics $\mu$. It follows that if one minimizes the loss (\ref{eq:L_defin}) and the resulting topics  satisfy the anchor words assumptions (this holds empirically, see below), then in the limit the topics must be the true topics.

The anchor words assumption was introduced in  \cite{anchor_words_nmf} as a part of a set of sufficient conditions for identifiability in NMF. 
A set $\Set{\mu_t}_{1}^T$ of topics is said to have an \emph{anchor words} property, denoted ($\mathcal{AW}$),  if for every $t\leq T$, there is a point $u_t \leq N$ such that $\mu_t(u_t) > 0$ and $\mu_{t'}(u_t) = 0$ for all $t' \neq t$. The point $u_t$ is called the anchor word of the topic $\mu_t$.  As mentioned earlier, natural topics tends to have the anchor word property. For instance, topics found by various LDA based methods have anchor words, \cite{Arora2}. We note that topics found by the FDM optimization algorithm, Algorithm \ref{alg:fdm_opt}, also have anchor words.

We first develop a few equivalent geometric characterizations 
of anchor words. While some of the arguments used in the proofs are similar in spirit to those in \cite{anchor_words_nmf},\cite{Arora1}, the particular notions we introduce have not appeared in the literature previously, and 
significantly clarify the geometric nature of the anchor word assumption. We thus provide these results here for completeness. Some of the equivalences below play an important role in the proof of Theorem \ref{thm:consistency}.

Denote by $\nsimplx$ the probability simplex in $\RR^N$.  A set of probability measures $\Set{\mu_t}_{1}^T \subset \nsimplx$ is called \emph{span-maximal} ($\mathcal{SM}$) if $\nsimplx \cap span \Set{\mu_t}_{1}^T = conv \Set{\mu_t}_{1}^T$.  Here $span$ and $conv$ denote the span and the convex hull of the set.  A set $\Set{\mu_t}_{1}^T \subset \nsimplx$ is \emph{positive} ($\mathcal{P}$) if the following holds: For every linear combination $\sum_{t} a_t$, 
if $\sum_{t} a_t \mu_t \in \nsimplx$ then $a_t\geq 0$ for all $t\leq T$. In other words, a set is positive if only its linear combinations with non-negative coefficients belong to the simplex.  Finally, a set is \emph{maximal} ($\mathcal{M}$) if the following holds: For any set $\Set{\nu_t}_{1}^T \subset \nsimplx$, if $conv \Set{\mu_t}_{1}^T \subset conv \Set{\nu_t}_{1}^T$, then we must have $\mu_t = \nu_{\pi(t)}$ for some permutation $\pi$ on $\Set{1,\dots,T} $. Equivalently, 
a set is maximal if it can not be properly contained in a convex hull of another set of topics of size $T$.

\begin{prop}
\label{lem:maximal_polytopes}
For a set of linearly independent topics $\Set{\mu_t}_{1}^T \subset \nsimplx$ , the properties ($\mathcal{SM}$), ($\mathcal{P}$), ($\mathcal{AW}$), and ($\mathcal{M}$) are equivalent. 
\end{prop}

The proof is given in supplementary material Section \ref{sec:supp_polytopes_proof}. Note in addition that $(\mathcal{AW})$ implies that 
$\Set{\mu_t}_{1}^T \subset \nsimplx$ are linearly independent. 
We now discuss the relation between between $\Set{\mu_t}_{1}^T$ and an FDM matrix $M_{u,v} = \sum_{i,j=1}^T \alpha_{ij} \mu_i(u) \mu_j(v) $. In particular, we describe how $\Set{\mu_t}_{1}^T$ can be recovered from the image of $M$ (as an operator, $M: \RR^N \rightarrow \RR^N$) under $(\mathcal{AW})$.
Note first that the image of $M$ satisfies $Im(M) \subseteq span \Set{\mu_t}_{1}^T$. Indeed, for a fixed $v$, a column of $M$ can be written as a linear combination,  $M_{\cdot,v} = \sum_i \Brack{ \sum_{j} \alpha_{i,j} \mu_j(v)} \mu_i (\cdot) $. Moreover, if the matrix $\alpha = \alpha_{ij}$ has full rank, $rank(\alpha) = T$, and if $\Set{\mu_t}_{1}^T$ are linearly independent, 
then $Im(M) = span \Set{\mu_t}_{1}^T$. Therefore, when $rank(\alpha) = T$, given $M$ we can recover $span \Set{\mu_t}_{1}^T$ as $Im(M)$.  Now, if $\Set{\mu_t}_{1}^T$ satisfies $(\mathcal{AW})$, then by Proposition \ref{lem:maximal_polytopes} it also satisfies $(\mathcal{SM})$. It then follows that $conv \Set{\mu_t}_{1}^T$ can be recovered as $conv \Set{\mu_t}_{1}^T = Im(M) \cap \nsimplx$. 
Finally, if we know $conv \Set{\mu_t}_{1}^T$, we can recover $\Set{\mu_t}_{1}^T$, since these are simply the extreme points of that polytope, and every polytope is uniquely characterized by its extreme points. This relation between
$\Set{\mu_t}_{1}^T$ and $M$ is at the basis of the consistency result below.

We state the consistency for the probabilistic setting: We complement the pLSA model to be a full generative model by assuming that topic distribution $\theta_d$ is sampled independently for each document $d$ from some probability distribution $\mathcal{T}$ on $\tsimplx$.  If $\mathcal{T}$ is a Dirichlet distribution, symmetric or asymmetric, this corresponds to an LDA model. Other examples include models with correlated topics, \cite{blei_corr}, or hierarchical topics, \cite{pachinko}, among many others. 
The only requirement on the topic distribution is the following: Let $\Theta_{i,j} = \Exp \theta_d(i) \theta_d(j)$ be the expected topic-topic co-occurrence matrix corresponding to the sampling scheme. We require that $\Theta$ is full rank. This assumption holds in all the examples above. 

\begin{thm}[Consistency]
\label{thm:consistency}
Consider a generative pLSA model (\ref{eq:intro_mu_d}), over topics $\Set{\mu_i}_{1}^T$ which satisfy ($\mathcal{AW}$), and  
where $\theta_d$ are sampled independently from a fixed distribution on $\tsimplx$, with topic-topic expected co-ocurrence matrix $\Theta$. 
Set $M_{u,v} := M(\mu,\Theta) = \sum_{i,j=1}^T \Theta_{i,j} \mu_i(u) \mu_j(v)$. 
Then with probability $1$, 
\begin{equation}
\label{eq:thm_direct_statement}
    \lim_{D \rightarrow \infty} \widehat{M}_{u,v} \log M_{u,v}  = \sum_{u,v} M_{u,v} \log M_{u,v}.
\end{equation}
Conversely, let $\Set{\mu'_t}_{t=1}^T \neq \Set{\mu_t}_{t=1}^T$ be a different set of topics satisfying ($\mathcal{AW}$).  There is a 
$\gamma = \gamma \Brack{\Set{\mu_i}_{1}^T,\Set{\mu'_i}_{1}^T,   \kappa(\Theta)}> 0$, such that 
for \emph{any} FDM $M'$ over $\Set{\mu'_i}_{1}^T$, with probability $1$,
\begin{equation}
\label{eq:thm_converse_statement}
    \lim_{D \rightarrow \infty} \widehat{M}_{u,v} \log M'_{u,v}  \leq  \sum_{u,v} M_{u,v} \log M_{u,v} - \gamma.
\end{equation}
\end{thm}

This result is proved in supplementary material Section \ref{sec:supp_thm_proof}. The probability in Theorem \ref{thm:consistency} is over the samples from the generative model, through the random variables $\whM$ (note that $\whM$, defined (\ref{eq:whm_formal_def}), depends on $D$). 
The theorem states that the loss (\ref{eq:L_defin}) is minimized at the true parameters $\Set{\mu_i}_{1}^T$ and $\Theta$. Indeed, denote $L_0 = \sum_{u,v} M_{u,v} \log M_{u,v}$. Note that the cost  $L(\mu,\Theta)$ is a random variable, as it depends on $\whM$. 
Then Eq.~(\ref{eq:thm_direct_statement}) states that $L(\mu,\Theta)$ converges to $L_0$, while 
Eq.~(\ref{eq:thm_converse_statement}) is equivalent to stating that for any other set of topics, 
 $\Set{\mu'_i}_{1}^T$, we have $L(\mu',\alpha)> L_0 + \gamma$ for any $\alpha$. The gap $\gamma$ will depend on how well $\Set{\mu'_i}_{1}^T$ approximates the true topics  $\Set{\mu_i}_{1}^T$.

\section{Experiments}
\label{sec:experiments}
\label{sec:nips_experiment}
\begin{figure*}
\centering
\subcaptionbox{
\label{fig:nips_log_like}
Distribution of log-likelihoods on test set. Larger (closer to zero) values are better. FDM (blue), SparseLDA (orange), EP (green).}{\includegraphics[width=0.48\textwidth]{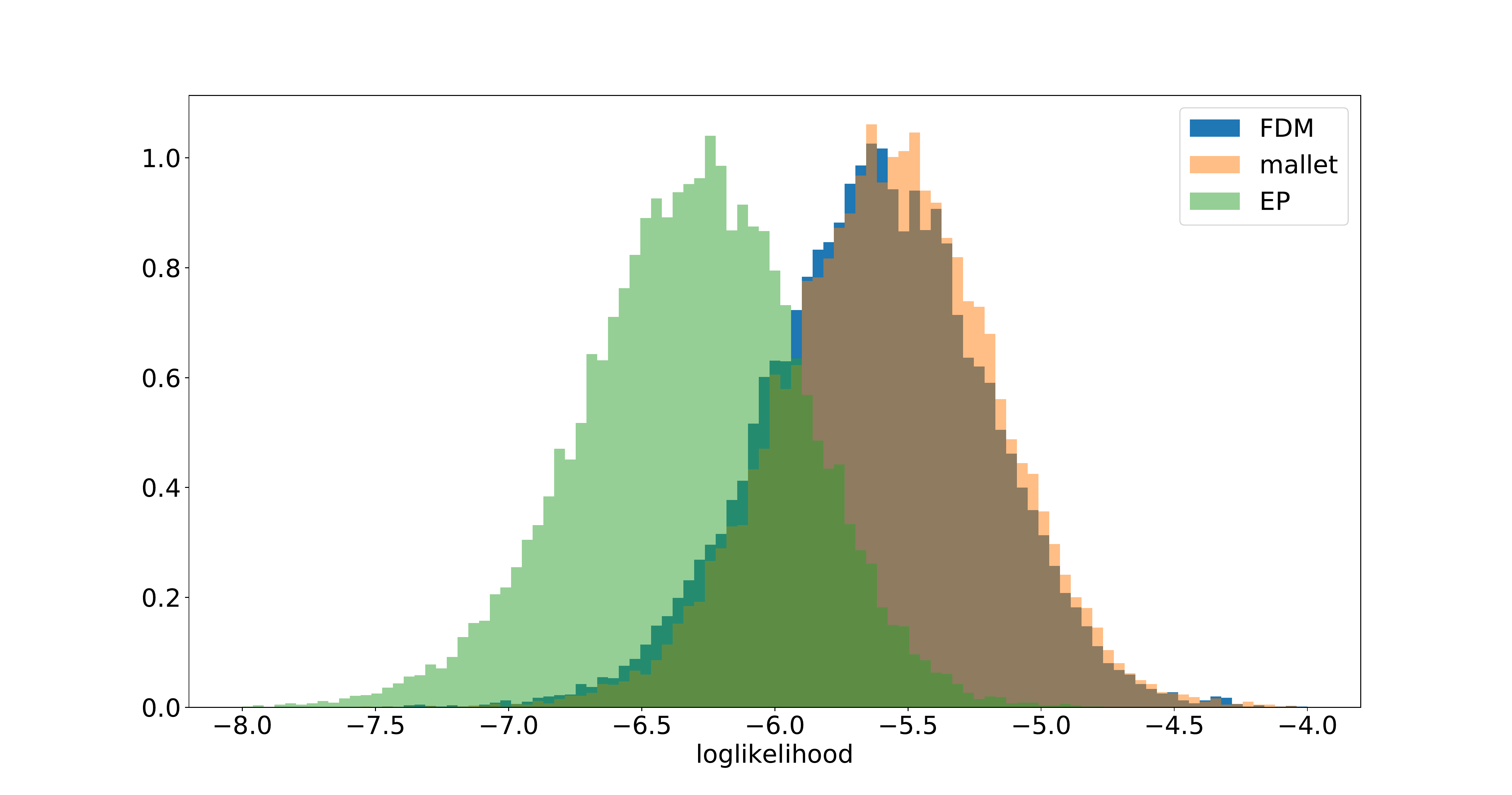}}%
\hfill 
\subcaptionbox{
\label{fig:nips_topic_matching}
Distribution of distances in an optimal Matching. FDM to SparseLDA (blue), FDM to empirical (orange), SparseLDA to empirical(green).}{\includegraphics[width=0.48\textwidth]{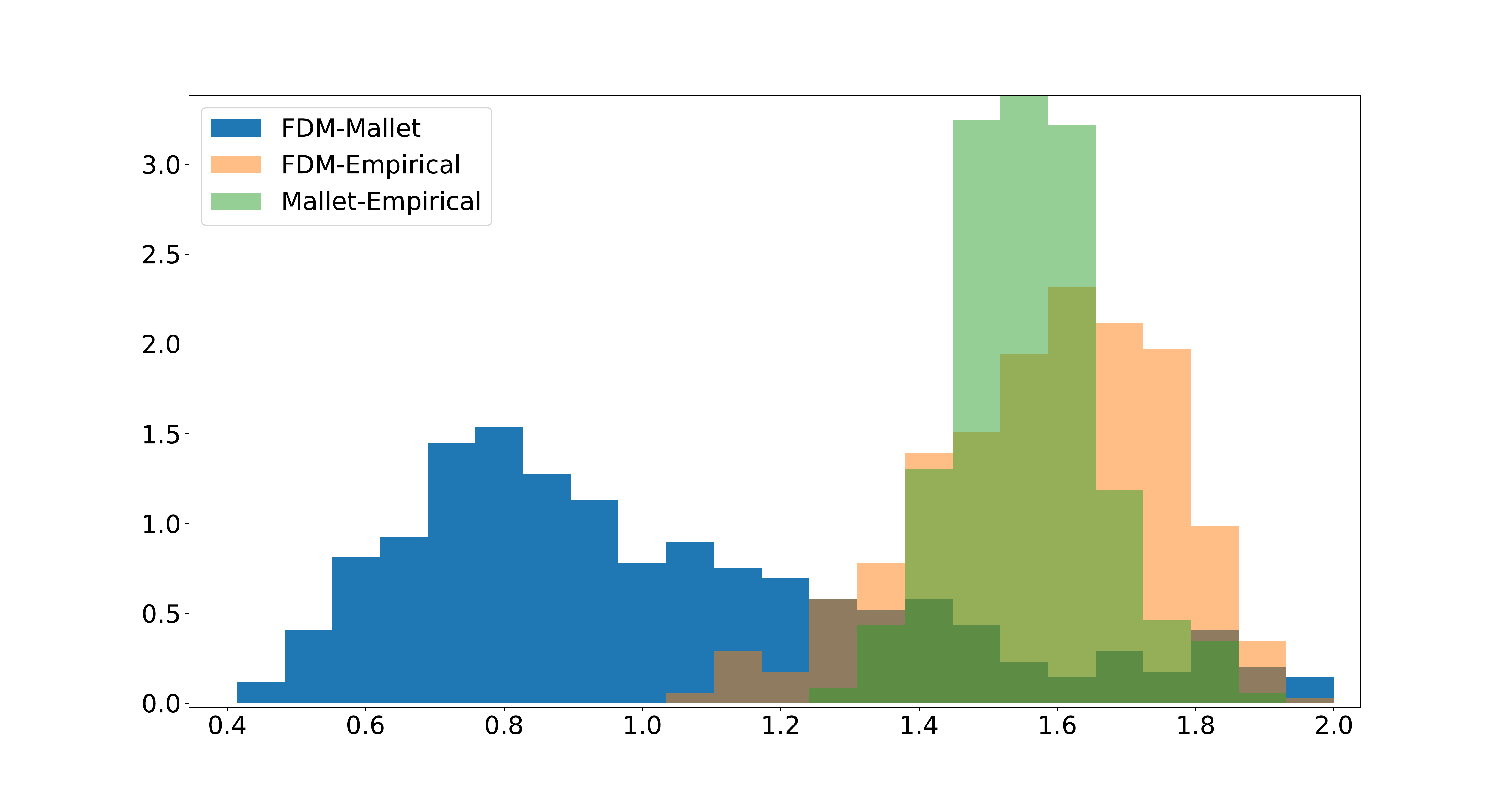}}%
\caption{NeurIPS Papers Experiment}
\end{figure*}

In the following sections we discuss experiments on  semi-synthetic data, the NeurIPS papers corpus and the Twitter dataset; see Section \ref{sec:subsec_exprimental} for an overview. 

\subsection{Synthetic Data}
\label{sec:expr_synthetic}

\begin{figure}
\centering
\includegraphics[width=\linewidth]{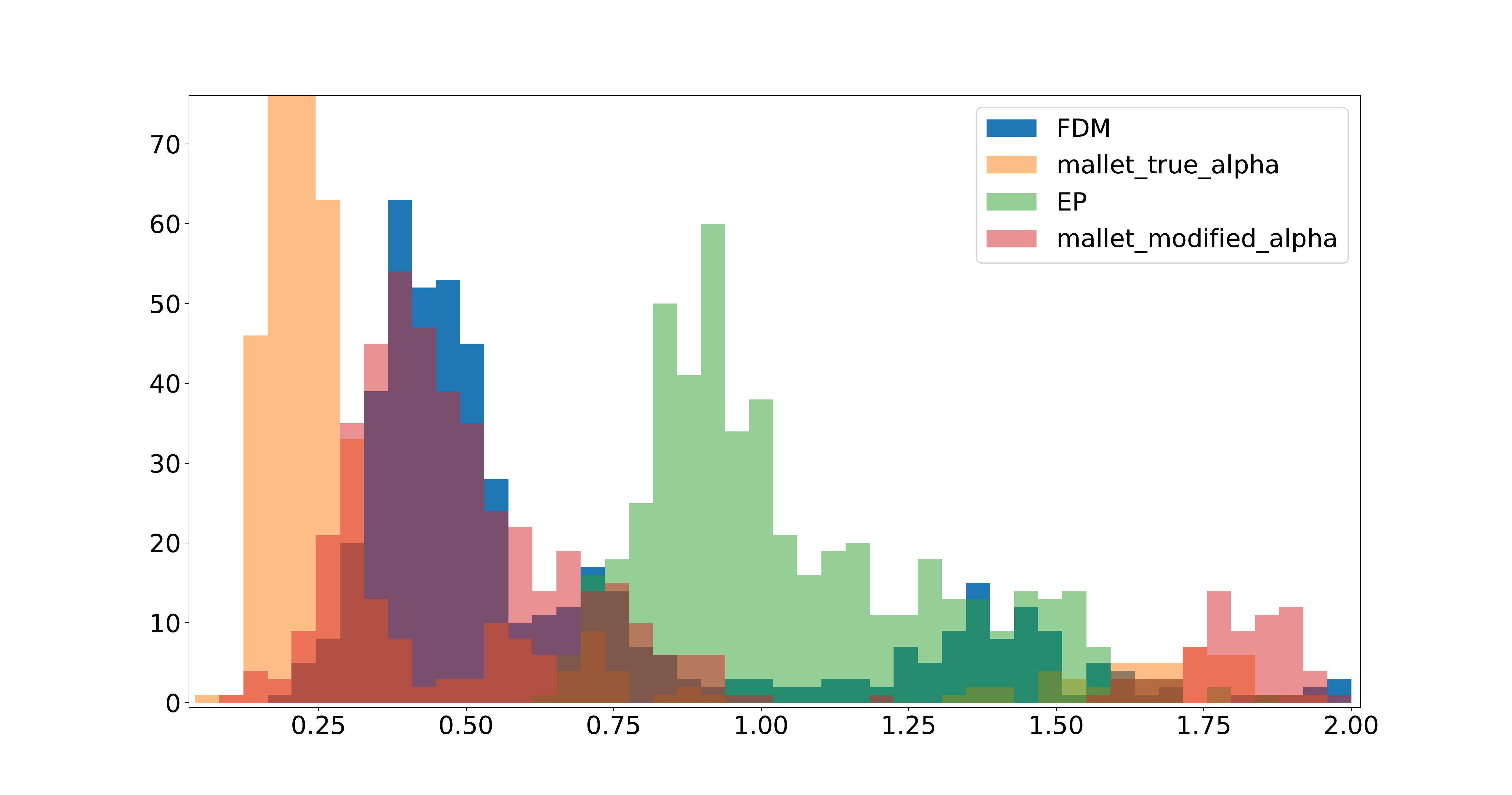}
\caption{Distribution of the distances to ground truths, in an optimal matching, for a typical instance of each method.}
\label{fig:synth_dist_hist}
\end{figure}

To approximate real data in a synthetic setting, we used $T=500$ topics learned by SparseLDA optimization on the NeurIPS papers corpus (Section \ref{sec:nips_experiment}) as the ground truth topics. The dictionary size in this case is $N=10500$ tokens. The synthetic documents were generated using the LDA model:
For each document, a topic distribution $\theta_d$ was sampled per document from a symmetric Dirichlet with the standard concentration parameter $\alpha = 1/T$, and $30$ tokens were sampled from the $\theta_d$ mixture of the ground truth topics. The corpus size was $D=100000$ documents. Note that for a dictionary of size $N=10500$ and non-uniform topics, this is not a  very large corpus.

To reconstruct the topics, we compared three algorithms: (i) SparseLDA -- a sparsity optimized parallel collapsed Gibbs sampler for LDA, implemented in the MALLET framework (see Section \ref{sec:literature} for details), which was run with 4 parallel threads. (ii) FDM, Algorithm \ref{alg:fdm_opt}.  (iii) The topic learning algorithm from \cite{Arora2}, to which we refer as EP (Extreme Points) in what follows.  SparseLDA was run with 4 threads. All algorithms were run 5 times, until convergence, on the fixed dataset. Note that the EP algorithm does not have a random initialization, but uses a random projection as an intermediate dimensionality reduction step. Therefore different runs might be affected by restarts (although very mildly, in practice). We used $M=1000$ as the dimension of the random projection -- a value that was specified in \cite{Arora2} as the practical value for the dictionary sizes of the order we use here. Hardware specifications are given in the supplementary material.
All the algorithms were run with the true number of topics $T$ as a parameter.

The SparseLDA algorithm was run in two modes: With the true hyperparameters, $\alpha = 1/T$, corresponding to the true $\alpha$ of the corpus, and with topic sparsity parameter $\beta = 1/N$, a standard setting which was also used to learn the ground truth topics. To model a situation where the true hyperparameters are unknown, we also evaluated SparseLDA with a modified hyperparameter $\alpha = 10/T$, and same $\beta$. Note that this is a relatively mild change of the hyperparameter.

The quality of the learned topics was measured by calculating the optimal matching $\ell_1$ \footnote{ $\Abs{\mu - \nu}_1 = \sum_{u\in \mathcal{X}} \Abs{\mu(u) -\nu(u)}$} distances to the ground truth topics. That is, given the topics returned by the model, $\nu_{t}$, $1\leq t \leq T$ and the ground truth topics $\mu_t$, we compute 
\begin{equation}
err = \frac{1}{T} \min_{\tau} \sum_{t=1}^T  \Abs{\mu_t - \nu_{\tau(t)}}_1,
\end{equation}
where $\tau$ is the matching --- a permutation of the set $\Set{1,\ldots,T}$. The optimal matching $\tau$ was computed using the Hungarian algorithm.

The results are given in Table \ref{tbl:synth_nips}, which shows for each algorithm the average error and the standard deviation over the different runs. To put the numbers 
in perspective, the typical $\ell_1$ distance between two ground truth topics is around $1.75$. Thus all algorithms learned at least some approximation of the ground truth. 
By visual inspection, a topic at distance $0.6$ from a given ground truth topic tends to capture the mass at the correct tokens, but the amount of mass deviates somewhat from the correct one.

We observe that SparseLDA with the ground truth $\alpha$ attained best performance. This is not surprising, since the algorithm is based on the generative model that is the true generative model of the corpus, and was provided with the true hyperparameters. Both of these constitute a considerable prior knowledge. 
The performance of the  EP algorithm was relatively low, which is likely due to the fact that the corpus size was not sufficiently large for that algorithm, with this set of topics.

Finally, FDM and SparseLDA with the modified hyperparameter attained similar performance. 
Interestingly, SparseLDA and FDM tend to err slightly differently. In Figure \ref{fig:synth_dist_hist} for a set of topics found by each algorithm we show the 
histogram of the quantities $\Abs{\mu_t - \nu_{\tau(t)}}_1$, i.e. the distribution of the distances within the matching. SparseLDA tends to completely miss about 50 to 80 out of 500 topics, while FDM is slightly less precise on the topics that it does approximate well. This figure is remarkably consistent across the different runs of the algorithms.

\begin{table}[h]
  \caption{Synthetic Corpus Matching Distances}
  \label{tbl:synth_nips}
  \centering
  \begin{tabular}{ll}
    \toprule
    Model  &  Average $\ell_1$ \\
    \midrule
    FDM            & $0.66 \pm 0.005$  \\
    SparseLDA, $\alpha = \alpha^*$ & $0.41 \pm 0.0004$  \\
    EP          & $1.05  \pm 0.005$ \\
    SparseLDA,  $\alpha = 10\alpha^*$  & $0.66 \pm 0.01$ \\
    \bottomrule
  \end{tabular}
\end{table}

\subsection{NeurIPS Dataset}
\label{sec:nips_experiment}
\begin{figure*}
\centering
\subcaptionbox{
\label{fig:twitter_log_like}
Distribution of log-likelihoods on test set. Larger (closer to zero) values are better. FDM (blue), SparseLDA (orange), EP (green).}{\includegraphics[width=0.48\textwidth]{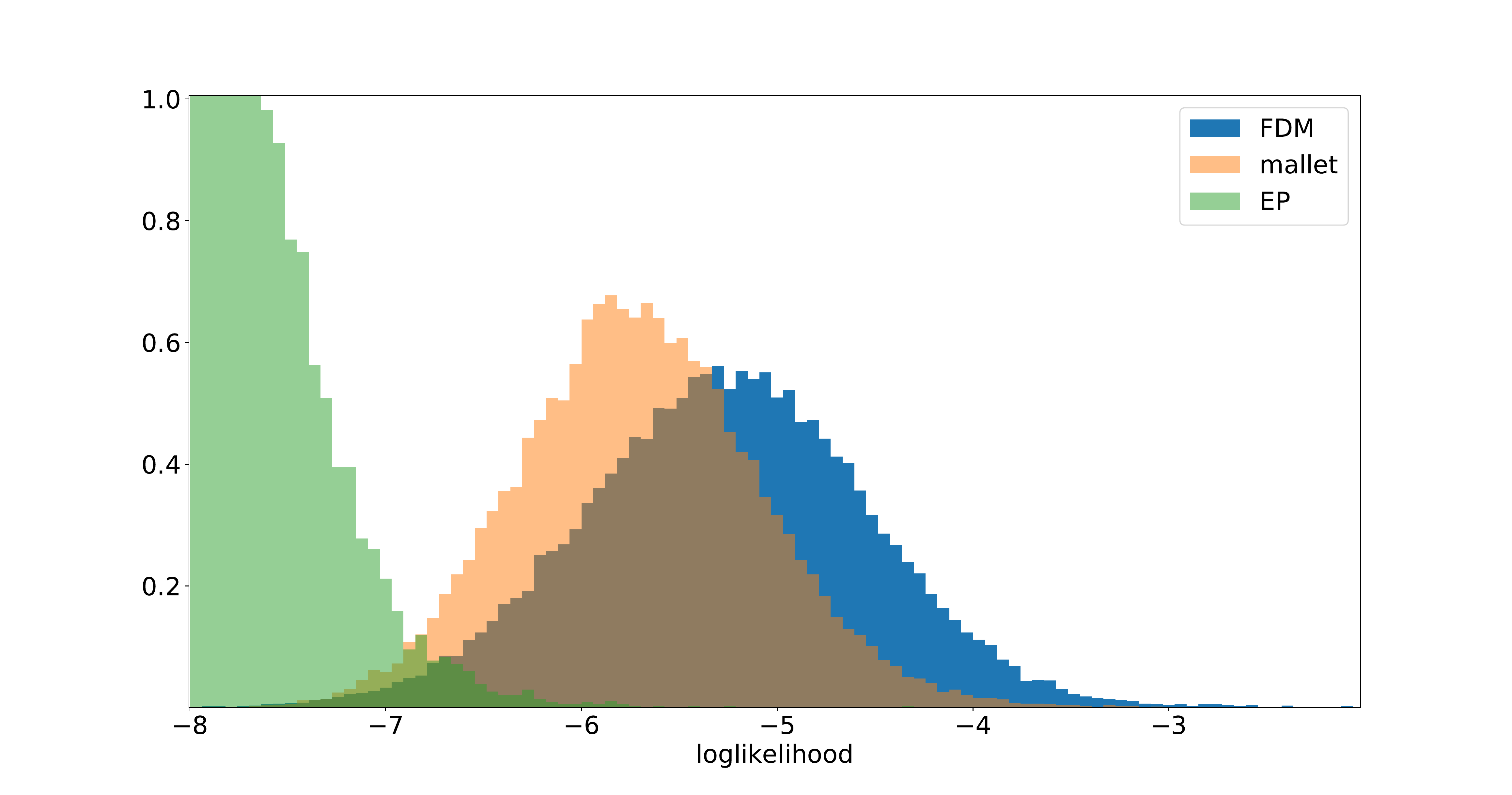}}%
\hfill 
\subcaptionbox{
\label{fig:twitter_topic_matching}
Distribution of distances in an optimal Matching. FDM to SparseLDA (blue), FDM to empirical (orange), SparseLDA to empirical(green).}{\includegraphics[width=0.48\textwidth]{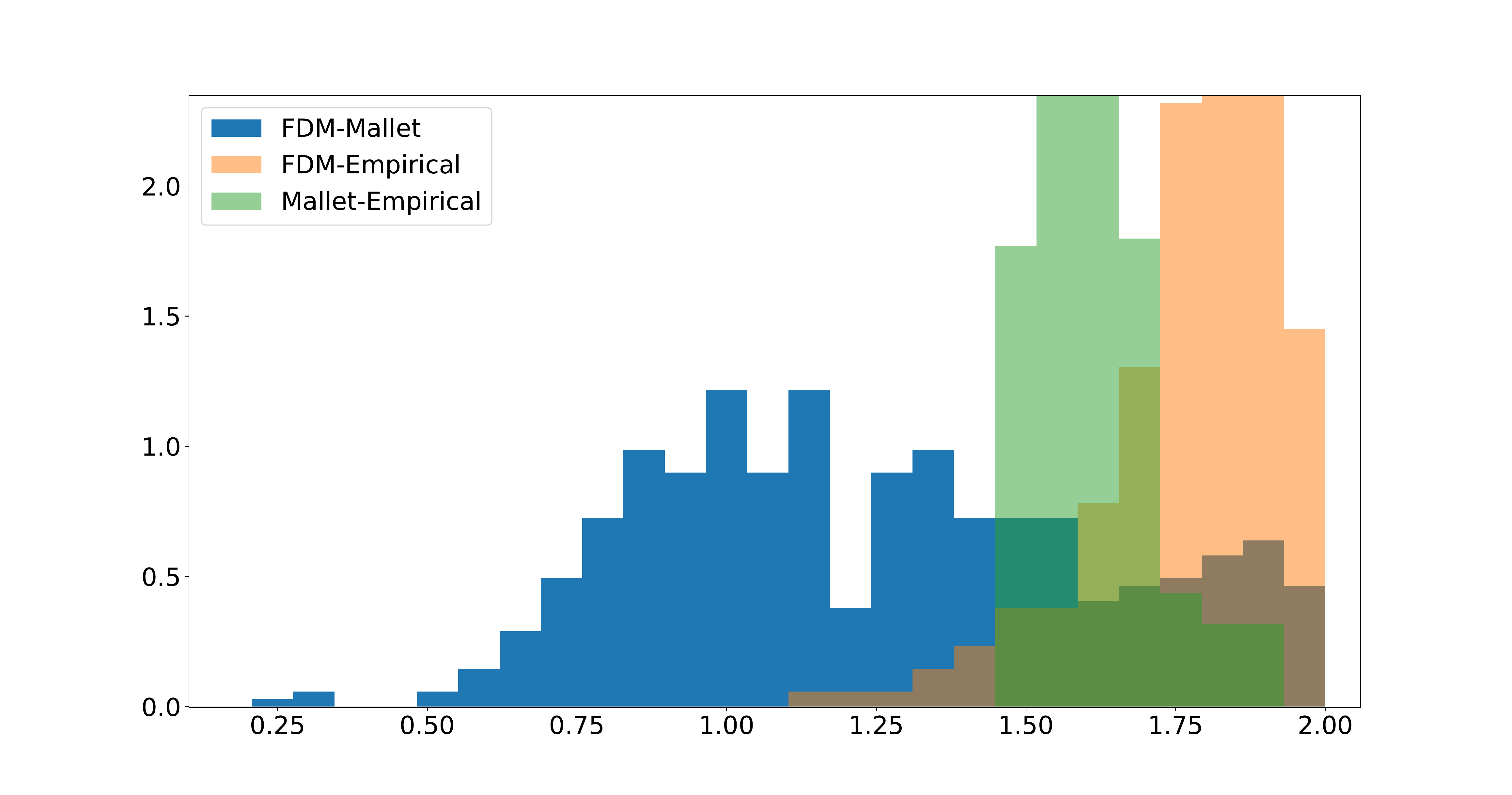}}%
\caption{Twitter Experiment}
\end{figure*}

The NeurIPS dataset, \cite{nips_corpus}, consists of 
all NeurIPS papers from 1987 to 2016. Each document was taken to be a single paragraph. Stop 
words, numbers and tokens appearing less than 50 times were removed. All tokens were 
stemmed and documents with less than 20 tokens are removed. The preprocessed dataset 
contained roughly $D = 251000$ documents over the dictionary of around $N = 10500$ unique tokens. $20\%$ of the 
documents were taken at random as a hold-out (test) set. 
The log-likelihood of the documents in the hold-out set was used as the performance measure. 
The computation of the (log) likelihood on the hold out set given topics is standard, with full details  provided in supplementary material Section \ref{sec:supp_likelihood_computation}.  The following models were trained: 
FDM, SparseLDA ($\alpha = 1/T, \beta = 1 / N$) and  EP (\cite{Arora2} ) with $T=500$ 
topics.  
All models were run until convergence.   

The mean hold-out log-likelihoods for each method are shown in Table  \ref{tbl:all_lglike}, and a histogram of the distribution of the holdout  log-likelihoods for a single run of each algorithm is shown in Figure \ref{fig:nips_log_like}.  We observe that the performance of SparseLDA and FDM are practically identical, and both perform better than EP.   

To obtain some insight into the relation between the models, Figure  \ref{fig:nips_topic_matching}  shows the histogram of the optimal matching distances between the topics learned by a fixed run of SpraseLDA and a fixed run of FDM (blue). For scale, the distances of SparseLDA and FDM to a fixed topic, the empirical distribution of the corpus, are also shown. It appears that both models find somewhat similar topics.  

\begin{table}[h]
  \caption{Test set average Log-likelihoods}
  \label{tbl:all_lglike}
  \centering
  \begin{tabular}{lll}
    \toprule
    Model  & NeurIPS & Twitter \\
    \midrule
    FDM & $-5.62$ & $-5.26$  \\
    SparseLDA & $-5.57$ & $-5.70$ \\
    EP  & $-6.31$ & $-8.5$ \\
    \bottomrule
  \end{tabular}
\end{table}

\subsection{Twitter Data}
\label{sec:expr-real-data-twitter}
We first describe the collection and processing of the Twitter 
corpus. The tweets were collected via the Tweeter API, 
\cite{twitter_api}.  The data contains about 16M (million) (after 
pre-processing, see below) publicly 
available tweets posted by 60K users during roughly the period 
of 1/2014 to 1/2017. The tweets were preprocessed to remove  
numbers, none-ASCII characters, mentions (Twitter usernames 
preceded by an @ character) and URLs were removed. All tokens 
were stemmed, and stop words and tokens with less than 3 
characters were removed. The most common 200 tokens and 
rare tokens (less than 1000 appearances in the corpus) were 
removed. Following this, tweets shorter than 4 tokens were also removed. This 
resulted in a corpus of about $D=16M$ tweets over a  dictionary of slightly more than $N=15000$ unique tokens. 
Each tweet was considered as a separate document, and typical tweets have 4 to 8 
tokens.  

The experiment setting was similar to the NeurIPS papers corpus. $20\%$ of the documents were taken as a hold-out set, and the holdout loglikelihood of SparseLDA, FDM and EP topics was evaluted. All algorithms were run to convergence, for about 12 hours for each run. 

The resulting hold-out log-likelihoods are given  in Table \ref{tbl:all_lglike}, and the distribution of the log-likelihoods in shown in Figure \ref{fig:twitter_log_like}.  We observe that in this case the performance of FDM is better than that of SparseLDA, while EP does not produce a good approximation of the dataset.

\section{Conclusions}
\label{sec:conclusion}
In this paper a introduced a new topic modeling approach, FDM topic modeling, which is based on 
matching, via KL divergence, of the token co-occurrence distribution induced by the topics to the 
co-occurrence distribution of the corpus. We have shown  the asymptotic consistency of the approach under the anchor words assumption and presented an efficient stochastic optimization 
procedure for this problem. This algorithm enables the approach to leverage GPU computation and efficient SGD optimizers. Our empirical evaluation shows that FDM produces topics of good quality and improves over the performance of SparseLDA.

\bibliographystyle{plain}
\bibliography{large_data_topics.bib}

\begin{thebibliography}{10}

\bibitem{parallel4}
Amr Ahmed, Moahmed Aly, Joseph Gonzalez, Shravan Narayanamurthy, and
  Alexander~J Smola.
\newblock Scalable inference in latent variable models.
\newblock In {\em Proceedings of the fifth ACM international conference on Web
  search and data mining}, pages 123--132. ACM, 2012.

\bibitem{tensor3}
Anima Anandkumar, Dean~P Foster, Daniel~J Hsu, Sham~M Kakade, and Yi~kai Liu.
\newblock A spectral algorithm for latent dirichlet allocation.
\newblock In {\em NIPS}. 2012.

\bibitem{Arora2}
Sanjeev Arora, Rong Ge, Yonatan Halpern, David Mimno, Ankur Moitra, David
  Sontag, Yichen Wu, and Michael Zhu.
\newblock A practical algorithm for topic modeling with provable guarantees.
\newblock In {\em ICML}, 2013.

\bibitem{Arora1}
Sanjeev Arora, Rong Ge, and Ankur Moitra.
\newblock Learning topic models -- going beyond svd.
\newblock FOCS '12, 2012.

\bibitem{bhatia97matrix}
R.~Bhatia.
\newblock {\em Matrix Analysis}.
\newblock Graduate Texts in Mathematics. Springer New York, 1997.

\bibitem{blei_corr}
David~M. Blei and John~D. Lafferty.
\newblock A correlated topic model of science.
\newblock {\em The Annals of Applied Statistics}, 2007.

\bibitem{Blei02latentdirichlet}
David~M. Blei, Andrew~Y. Ng, and Michael~I. Jordan.
\newblock Latent dirichlet allocation.
\newblock {\em Journal of Machine Learning Research}, 3, 2002.

\bibitem{coverthomas}
Thomas~M. Cover and Joy~A. Thomas.
\newblock {\em Elements of Information Theory (Wiley Series in
  Telecommunications and Signal Processing)}.
\newblock Wiley-Interscience, 2006.

\bibitem{anchor_words_nmf}
David Donoho and Victoria Stodden.
\newblock When does non-negative matrix factorization give a correct
  decomposition into parts?
\newblock In {\em Proceedings of the 16th International Conference on Neural
  Information Processing Systems}, NIPS’03, 2003.

\bibitem{variational1}
James Foulds, Levi Boyles, Christopher DuBois, Padhraic Smyth, and Max Welling.
\newblock Stochastic collapsed variational bayesian inference for latent
  dirichlet allocation.
\newblock In {\em KDD}, 2013.

\bibitem{cgibbs}
Thomas~L. Griffiths and Mark Steyvers.
\newblock {Finding scientific topics}.
\newblock {\em Proceedings of the National Academy of Sciences of the United
  States of America}, 101, 2004.

\bibitem{online2}
N.~{Guan}, D.~{Tao}, Z.~{Luo}, and B.~{Yuan}.
\newblock Online nonnegative matrix factorization with robust stochastic
  approximation.
\newblock {\em IEEE Transactions on Neural Networks and Learning Systems},
  2012.

\bibitem{streaming_lda}
Matthew Hoffman, Francis~R. Bach, and David~M. Blei.
\newblock Online learning for latent dirichlet allocation.
\newblock In {\em Advances in Neural Information Processing Systems 23}. 2010.

\bibitem{variational2}
Matthew~D. Hoffman, David~M. Blei, Chong Wang, and John Paisley.
\newblock Stochastic variational inference.
\newblock {\em J. Mach. Learn. Res.}, 14(1), 2013.

\bibitem{PLSA}
Thomas Hofmann.
\newblock Probabilistic latent semantic indexing.
\newblock SIGIR '99, 1999.

\bibitem{nmf_survey}
Zhengyu Huang, Aimin Zhou, and Guixu Zhang.
\newblock Non-negative matrix factorization: A short survey on methods and
  applications.
\newblock In {\em Computational Intelligence and Intelligent Systems}. Springer
  Berlin Heidelberg, 2012.

\bibitem{adam}
Diederik~P. Kingma and Jimmy Ba.
\newblock Adam: A method for stochastic optimization.
\newblock In {\em International Conference on Learning Representations}, 2015.

\bibitem{walker_alias}
Aaron~Q. Li, Amr Ahmed, Sujith Ravi, and Alexander~J. Smola.
\newblock Reducing the sampling complexity of topic models.
\newblock In {\em KDD}, 2014.

\bibitem{pachinko}
Wei Li and Andrew McCallum.
\newblock Pachinko allocation: Dag-structured mixture models of topic
  correlations.
\newblock In {\em Proceedings of the 23rd International Conference on Machine
  Learning}, ICML ’06. Association for Computing Machinery, 2006.

\bibitem{parallel2}
Zhiyuan Liu, Yuzhou Zhang, Edward~Y. Chang, and Maosong Sun.
\newblock Plda+: Parallel latent dirichlet allocation with data placement and
  pipeline processing.
\newblock {\em ACM Trans. Intell. Syst. Technol.}, 2011.

\bibitem{McCallumMALLET}
Andrew~Kachites McCallum.
\newblock Mallet: A machine learning for language toolkit.
\newblock http://mallet.cs.umass.edu, 2002.

\bibitem{nips_corpus}
NeurIPSPapersCorpus.
\newblock Neuripspaperscorpus.
\newblock \url{https://www.kaggle.com/benhamner/nips-papers}, 2016.
\newblock Accessed: 1/8/2019.

\bibitem{mallet1}
David Newman, Arthur Asuncion, Padhraic Smyth, and Max Welling.
\newblock Distributed algorithms for topic models.
\newblock {\em J. Mach. Learn. Res.}, 2009.

\bibitem{scikit-learn}
F.~Pedregosa, G.~Varoquaux, A.~Gramfort, V.~Michel, B.~Thirion, O.~Grisel,
  M.~Blondel, P.~Prettenhofer, R.~Weiss, V.~Dubourg, J.~Vanderplas, A.~Passos,
  D.~Cournapeau, M.~Brucher, M.~Perrot, and E.~Duchesnay.
\newblock Scikit-learn: Machine learning in {P}ython.
\newblock {\em Journal of Machine Learning Research}, 12:2825--2830, 2011.

\bibitem{hogwild}
Benjamin Recht, Christopher Re, Stephen Wright, and Feng Niu.
\newblock Hogwild: A lock-free approach to parallelizing stochastic gradient
  descent.
\newblock In {\em Advances in Neural Information Processing Systems 24}. 2011.

\bibitem{smola_parallel}
Alexander Smola and Shravan Narayanamurthy.
\newblock An architecture for parallel topic models.
\newblock {\em Proc. VLDB Endow.}, 2010.

\bibitem{cgibbs_book}
Mark Steyvers and Tom Griffiths.
\newblock {\em {Probabilistic Topic Models}}.
\newblock 2007.

\bibitem{tristan}
Jean-Baptiste Tristan, Joseph Tassarotti, and Guy Steele.
\newblock Efficient training of lda on a gpu by mean-for-mode estimation.
\newblock In {\em ICML}, 2015.

\bibitem{twitter_api}
TwitterAPI.
\newblock Twitter api.
\newblock
  \url{https://developer.twitter.com/en/docs/tweets/timelines/api-reference/get-statuses-user_timeline.html},
  2019.
\newblock Accessed: 1/8/2019.

\bibitem{cgibbs_sparcity}
Han Xiao and Thomas Stibor.
\newblock Efficient collapsed gibbs sampling for latent dirichlet allocation.
\newblock In {\em ACML}, 2010.

\bibitem{mallet2}
Limin Yao, David Mimno, and Andrew McCallum.
\newblock Efficient methods for topic model inference on streaming document
  collections.
\newblock In {\em KDD}, 2009.

\bibitem{parallel3}
Ke~Zhai, Jordan Boyd-Graber, Nima Asadi, and Jude Khouja.
\newblock Mr. lda: a flexible large scale topic modeling package using
  variational inference in mapreduce.
\newblock In {\em WWW}, 2012.

\bibitem{online1}
R.~{Zhao} and V.~Y.~F. {Tan}.
\newblock Online nonnegative matrix factorization with outliers.
\newblock In {\em ICASSP}, 2016.

\end{thebibliography}

\appendix 
\newpage 
\onecolumn
{\LARGE \centering{\textbf{Topic Modeling via Full Dependence Mixtures - Supplementary Material}}} 
\newline

\section{Outline}
In this supplementary material we provide the proofs of the results stated in the main text. In addition, details on the hardware used in the experiments are given in Section \ref{sec:hardware}.

\section{The Unbiased Estimate $\whM_d$}
\label{sec:supp_unbiased_estimates}
In this section we  discuss the construction of the matrix $\whM_d$, defined in (\ref{eq:whM_d_alg_definition}) in Algorithm \ref{alg:fdm_calc}. Specifically, we show the unbiased estimate property of 
$\whM_d$, (\ref{eq:unbiased_estimate_main_text}). 

First, let us introduce some notation. 
For any two vectors $a,b \in \RR^N$, denote by $a \otimes b \in \RR^{N\times N}$ the outer product of $a,b$, an $N\times N$ matrix given by:
\begin{equation}
\Brack{a \otimes b}_{u,v} = a(v) \cdot b(u). 
\end{equation}
For any probability distribution $\mu$ on $\mathcal{X}$, $(\mu \otimes \mu)_{u,v}$ is simply the 
probability of obtaining the pair $u,v$ when sampling independently twice from $\mu$.

For a document $d$, recall that $c_d$, defined in (\ref{eq:c_d_def}), is the token counts vector
of the document $d$ and $l_d$ is the total number of tokens in $d$. Set $\hat{d} = \frac{1}{l_d} c_d$ to be the empirical probability distribution on $\mathcal{X}$ corresponding to $d$. 

As described in the main text,  assuming the pLSA 
model, each document $d$ is an i.i.d sample from 
some mixture of topics 
\begin{equation}
\nu_d = \sum_t \theta_d(t) \mu_t.
\end{equation}

Let us fix some mixture
\begin{equation}
\nu = \sum_t \theta(t) \mu_t.
\end{equation}

The co-occurrence matrix for the mixture is by definition 
\begin{equation}
    \Brack{M_{\nu}}_{u,v} = \nu(u)\cdot \nu(v) = 
       \sum_{i,j\leq T} \theta(i)\theta(j) \mu_i(u) \mu_j(v). 
\end{equation}
Note that with our notation, we have 
\begin{equation}
    M_{\nu} = \nu \otimes \nu.
\end{equation}

Moreover, given a document $d$, note that the 
empirical co-ocurrence matrix $\whM_d$ may be written as 
\begin{equation}
\whM_d = \frac{l_d}{l_d -1 }\hat{d}\otimes \hat{d} - \frac{1}{l_d-1}Diag(\hat{d}).
\end{equation}
Here $Diag(x)$ is a diagonal $N \times N$ matrix with diagonal entries given by the vector 
$x\in \RR^N$.

We first compute the expectation of $\hat{d} \otimes \hat{d}$ in the following Lemma. 
\begin{lem}
\label{lem:moment_estimate}
   Let $d = \Set{x_1, \ldots, x_{l_d}}$ be an i.i.d sample from a distribution $\nu$. Then 
  \begin{equation}
     \Exp{ \hat{d} \otimes \hat{d}} = \frac{l_d-1}{l_d} \nu \otimes \nu + \frac{1}{l_d} Diag(\nu).
  \end{equation}
\end{lem}
\begin{proof}
    Consider the coordinate $u,v$ of the matrix $\Exp{ \hat{d} \otimes \hat{d}}$. 
    \begin{eqnarray*}
        \Brack{\Exp{ \hat{d} \otimes \hat{d}}}_{u,v} = \Exp{ \hat{d}(u) \cdot \hat{d}(v)} = \\
        \frac{1}{{}l_d^2}\Exp{ \Brack{\sum_{i=1}^{l_d} \Ind{x_i = u} } \cdot 
              \Brack{\sum_{j=1}^{l_d} \Ind{x_j = v} }
            }  = \\
       \frac{1}{{l_d}^2}\sum_{i,j=1}^{l_d} 
        \Exp{\Brack{ \Ind{x_i = u}\cdot \Ind{x_j = v}} }.  
    \end{eqnarray*}
The  results is now obtained by by considering separately the cases $i=j$, $i\neq j$, $u=v$, $u\neq v$.  Indeed, choose for instance fixed $u\neq v$.  For $i\neq j$ we have
\begin{equation}
    \Exp{\Brack{ \Ind{x_i = u}\cdot \Ind{x_j = v}} } = \nu(u) \cdot \nu(v). 
\end{equation}
For $i = j$, since $u \neq v$,
\begin{equation}
    \Exp{\Brack{ \Ind{x_i = u}\cdot \Ind{x_i = v}} } = 0.
\end{equation}
Since there are $l_d^2 - l_d$ pairs $i,j$ with $i \neq j$, we thus have overall that 
\begin{equation}
   \Brack{\Exp{ \hat{d} \otimes \hat{d}}}_{u,v} =  \frac{1}{{l_d}^2}\sum_{i,j=1}^{l_d} 
        \Exp{\Brack{ \Ind{x_i = u}\cdot \Ind{x_j = v}} } = \frac{l_d(l_d -1)}{l_d^2} \nu(u) \cdot \nu(v) 
        = \frac{l_d -1}{l_d} \nu(u) \cdot \nu(v)  = \Brack{\frac{l_d -1}{l_d} \nu \otimes \nu}_{u,v}.
\end{equation}
The diagonal case, $u=v$, is handled similarly.
\end{proof}

It follows therefore that if $d$ is a document constructed by sampling $l_d$ tokens from $\nu$, then 
$\hat{d}\otimes \hat{d}$ is not an unbiased estimate of $\nu \otimes \nu$. One can however easily fix this 
by subtracting the diagonal and renormalizing. Indeed, from Lemma \ref{lem:moment_estimate} we have
\begin{eqnarray}
\label{eq:unbiased_moment_estimate}
      \nu \otimes \nu &=& \frac{l_d}{l_d-1} \Exp{\hat{d} \otimes \hat{d} }  -     
      \frac{1}{l_d-1} Diag(\nu)   \\
      &=& \frac{l_d}{l_d-1} \Exp{\hat{d} \otimes \hat{d} }  -  
      \frac{1}{l_d-1} \Exp{Diag(\hat{d})}  \nonumber \\
      &=& \Exp{ \Brack{\frac{l_d}{l_d-1} \hat{d} \otimes \hat{d}   -  \frac{1}{l_d-1} Diag(\hat{d})}}  \nonumber \\
      &=& \Exp \whM_d.  \nonumber
\end{eqnarray}

That is, $\whM_d$ is an unbiased estimate of $\nu_d \otimes \nu_d$.

\section{Proof of Proposition \ref{lem:maximal_polytopes}}
\label{sec:supp_polytopes_proof}
\begin{proof}
$(sM) \implies (P)$: Suppose $\mu = \sum_{t} a_t \mu_t \in \nsimplx$. 
Since also have $\mu \in span \Set{\mu_t}_{1}^T$,  by span-maximality 
it follows that there are $b_t \geq 0$, with $\sum_t b_t = 1$ such that $\mu = \sum b_t \phi_t$. Now, by linear independence, there is a unique representation of $\mu$. Thus $a_t = b_t$. Conversely, to show $(P) \implies (sM)$, suppose $\mu  = \sum_t a_t \mu_t \in \nsimplx \cap span \Set{\mu_t}_{1}^T$. By (P) we have $a_t \geq 0$. Note also that $1 = \sum_u \mu(u) = \sum_t \sum_u a_t \mu_t(u) = \sum_t a_t$. Thus $\mu \in conv \Set{\mu_t}_{1}^T$. 

We now show $(P) \iff (\mathcal{AW})$. First assume $(\mathcal{AW})$. Let $u_t$ be the anchor words. Choose some $\mu = \sum_{t} a_t \mu_t \in \nsimplx$. Then by definition of the anchor words, $\mu_t(u_t) = \sum_{t} a_t \mu_t(u_t) = a_t \mu_t(u_t)$. Since $\mu_t(u_t) >0$, and $\mu \in \nsimplx$ this implies $a_t \geq 0$. For the converse, we show $\neg (\mathcal{AW}) \implies \neg (P)$. Assume $\neg (\mathcal{AW})$. This means that there is a topic $\mu_t$ such that for every token $u$ for which $\mu_t(u)>0$, there is another topic, $\mu_{t_u}$, such that $\mu_{t_u}(u) > 0$. In particular, set $\tilde{\mu} = \frac{1}{T-1} \sum_{t' \neq t} \mu_{t'}$. It follows then that $\mu_t$ is absolutely continuous with respect to $\mu$. That is, for every $u$ for which $\mu_t(u) >0$, we have $\tilde{\mu}(u) > 0$. Using this property, it is clear that for a small enough $\eps >0$, we have $\tilde{\mu}(u) - \eps \mu_t(u) \geq 0$ for all $u$. Thus $\frac{1}{1-\eps} \Brack{\tilde{\mu}(u) - \eps \mu_t(u)} \in \nsimplx$. However, this expression has a strictly negative coefficient at $\mu_t$, thus contradicting $(P)$. 

Finally, we show $(sM) \implies (M)$ and  $(M) \implies (AW)$.

$(sM) \implies (M)$: \\
Assume $(sM)$. We show $(M')$. Since $conv \Set{\mu_t}_{1}^T \subset conv \Set{\nu_t}_{1}^T$ and since $\mu_t$ are linearly independent, we have  $span \Set{\mu_t}_{1}^T = span \Set{\nu_t}_{1}^T$. By $(sM)$, this implies $conv \Set{\mu_t}_{1}^T = conv \Set{\nu_t}_{1}^T$. Since any polytope defines its extreme points uniquely, we obtain the conclusion of $(M')$.

For $(M) \implies (AW)$, we equivalently show $\neg (\mathcal{AW}) \implies \neg (M')$. Assume $\neg (\mathcal{AW})$. Consider the 
distribution $\widetilde{\mu_t} = \frac{1}{1-\eps} \Brack{\tilde{\mu}(u) - \eps \mu_t(u)}$ constructed earlier. 
Clearly we may write $\mu_t = \alpha \widetilde{\mu_t} + \beta \tilde{\mu} $  with $\alpha+\beta =1$, and  $\alpha,\beta >0$, strictly positive. It follows that $conv \Set{\mu_t}_{1}^T \subsetneq conv \Brack{\Set{\widetilde{\mu_t}} \cup \Set{\mu_{t'}}_{t' \neq t}}$, thus implying $(M')$.
\end{proof}

\section{Proof of Theorem \ref{thm:consistency}}
\label{sec:supp_thm_proof}
In this section we use the notation introduced in Section \ref{sec:supp_unbiased_estimates}.

\begin{proof}
It follows from Lemma \ref{lem:moment_estimate} and (\ref{eq:unbiased_moment_estimate}) that conditioned on $\theta_d$, we have 
$\Exp{ \widehat{M_d} |\; \theta_d}= \sum_{i,j=1}^T \theta_d(i) \theta_d(j) \mu_i \otimes \mu_j$. 
Therefore by definition, 
\begin{equation}
\Exp{ \widehat{M_d}}  = \Exp {\sum_{i,j=1}^T \theta_d(i) \theta_d(j) \mu_i \otimes \mu_j} =  \sum_{i,j=1}^T \Theta_{i,j} \mu_i(u) \mu_j(v) = M. 
\end{equation}
Since $M_d$ are bounded independent variables in $\RR^{N\times N}$, and 
$\widehat{M} = \frac{1}{D} \Brack{\sum_{i=1}^D
 \widehat{M}_{d_i}   }$, 
by the Law of Large Numbers we have $\widehat{M} \rightarrow M$ with probability $1$, which in turn implies (\ref{eq:thm_direct_statement}).

Conversely, let 
\begin{equation}
\label{eq:Mprime_form_thm}
M' = \sum_{i,j} \beta_{i,j} \mu'_i \otimes \mu'_j 
\end{equation}
be an FDM based on the topics $\Set{\mu'_i}_{1}^T$. 
Set $V = span \Set{\mu_i}_{1}^T$ and $V' = span \Set{\mu'_i}_{1}^T$. Clearly $Im(M') \subseteq V'$. 
On the other hand, since $\Theta$ has full rank, 
$\kappa(\Theta)>0$, we have 
$Im(M) = V$ and  $\kappa(M) > 0$. 
Next, by Lemma \ref{lem:maximal_polytopes}, topics with ($\mathcal{AW}$) property are identified by their span.  
Since $\Set{\mu'_i}_{1}^T \neq \Set{\mu_i}_{1}^T$, this implies $V \neq V'$. 
Set 
\begin{equation}
\label{eq:gamma_def_thm}
 \gamma' := \inf_{M' \; s.t. \; Im(M') \subseteq V'} \norm{M - M'}_{op}  > 0,
\end{equation}
where the norm is (say) the operator norm. The fact that $\gamma > 0$ is crucial and follows, for instance, from the 
Davis-Kahan  ``sin'' Perturbation Theorem, 
\cite{bhatia97matrix}(Section VII.3). Indeed, assume to the contrary that $\gamma = 0$. This would imply that there is a sequence $M'_n$, with $Im(M'_n) \subseteq V'$,  converging to $M$. By Davis-Kahan Theorem, the eigen-spaces of $M'_n$ must converge to those of $M$. However, since $\kappa(M) >0$ and since $V \neq V'$  (and hence $\sin(V,V') > 0$, see \cite{bhatia97matrix}), this is impossible. 

It remains to observe that since all finite dimensional norms are equivalent, $\gamma > 0$ implies $\norm{M - M'}_1 \geq \gamma'' > 0$ where 
$\norm{\cdot}_1$ is the coordinatewise $\ell_1$ norm. Next, Pinsker's Inequality, \cite{coverthomas},  yields a lower bound on the KL-divergence,
\begin{equation}
  0 < \gamma'' \leq \norm{M - M'}_1 \leq \sqrt{2 D_{KL}(M | M')}, 
\end{equation}
where 
\begin{align}
    D_{KL}(M | M') &= \sum_{u,v} M_{u,v} \log \frac{M_{u,v}}{M'_{u,v}} \\ 
    &= \sum_{u,v} M_{u,v} \log M_{u,v} - 
    \sum_{u,v} M_{u,v} \log M'_{u,v}. 
    \label{eq:kl_expanded}
\end{align}
To summarize, 
\begin{equation}
\label{eq:gamma_def_thm_kl}
 \gamma''' := \inf_{M' \; s.t. \; Im(M') \subseteq V'} D_{KL}(M'|M)  > 0. 
\end{equation}
Finally, we obtain
\begin{align}
    \lim_{D\rightarrow \infty} 
    \sum_{u,v} \widehat{M}_{u,v} \log(M'_{u,v})  
    &= \sum_{u,v} M_{u,v} \log(M'_{u,v})  \label{eq:cons_thm_final_line1}\\
    &\leq \sum_{u,v} M_{u,v} \log(M_{u,v}) - \gamma''', 
    \label{eq:cons_thm_final_line2}
\end{align}
where (\ref{eq:cons_thm_final_line1}) follows from the first part of the proof and (\ref{eq:cons_thm_final_line2}) follows form 
(\ref{eq:gamma_def_thm_kl}) and (\ref{eq:kl_expanded}). 
\end{proof}

\section{Holdout Likelihood Computation}
\label{sec:supp_likelihood_computation}

For a document $d$ given as a sequence 
of tokens $d = \Set{x_1,\ldots,x_{l_d}}$, where $l_d$ is the total number of tokens in $d$, 
recall from the main text that we denote by  
$c_d \in \RR^N$ the count vector of $d$, \begin{equation}
\label{eq:c_d_def}
c_d(u) = \#\Set{ x_i \setsep x_i = u } \mbox{ for $u \in \mathcal{X}$}. 
\end{equation}
The empirical distribution of $d$, denoted 
by $\hat{d}$ is the normalized count vector, 
$\hat{d} = \frac{1}{l_d} \cdot c_d \in \RR^N$.

Given the topics returned by the model, $\Set{\mu_i \mid i=1,..,T}$ for every  document 
$d$ we first compute the topics assignment $\theta_d$ as follows:
\begin{eqnarray}
    \label{eq:optimal_kl_assignment}
    \theta_d &=& \argmin_{\theta} D_{KL}(\hat{d}, m(\theta)) \\ 
    &=&  \argmax_{\theta}  \sum_{u\in \mathcal{X}} \hat{d}(u) \log \BBrack{m(\theta)(u)} \\
    &=& \argmax_{\theta}  \frac{1}{l_d}\sum_{x_i \in d} \log \BBrack{m(\theta)(x_i)}
    ,
\end{eqnarray}

where $m(\theta)$ is the mixture generated by the topics and the assignment $\theta$, 
$m(\theta) = \sum_{t} \theta(t) \mu_t$, 
and $D_{KL}$ is the Kullback-Leibler 
divergence. Thus $\theta_d$ is the assignment such that the mixture $m(\theta_d)$ best 
approximates the document in KL divergence. Equivalently, $\theta_d$ is the assignment such that the 
mixture $m(\theta_d)$ gives the highest likelihood to the document $d$. This is the standard definition 
of likelihood for pLSA models.

Note that (\ref{eq:optimal_kl_assignment}) 
is a convex problem in $\theta$, and can be solved efficiently and in parallel over the 
documents. Solving (\ref{eq:optimal_kl_assignment}) is a standard step in most 
in pLSA and Non-Negative Matrix Factorization methods, and existing efficient implementations
may be used. See for instance \cite{scikit-learn}.

Next, given $\theta_d$ we compute the document likelihood $L_d$ as $L_d = \sum_{u \in \mathcal{X}} \hat{d}(u) \log m(\theta_d)(u)$ and take the overall likelihood of the holdout set of documents to be the average of $L_d$ over all 
documents in the set.

\section{Hardware Specifications}
\label{sec:hardware}

 The LDA models were trained using an Intel Core i7-6950X processor, using the collapsed Gibbs sampler 
algorithm implemented in the MALLET package, \cite{McCallumMALLET}.
The FDM models were trained using NVIDIA GeForce GTX 1080 Ti GPU, and the optimization was implemented using TensorFlow 1.9.0 and Adam SGD optimizer with $0.001$ learning rate.

\end{document}